\documentclass[num-refs]{wiley-article}
\usepackage[utf8]{inputenc}
\usepackage{lipsum}
\usepackage{booktabs}
\PassOptionsToPackage{hyphens}{url}
\usepackage{hyperref}
\usepackage{xcolor}
\hypersetup{
    colorlinks,
    linkcolor={red!50!black},
    citecolor={blue!50!black},
    urlcolor={black}
}
\newcommand*{\restartrowcolors}{%
  \ifhmode\unskip\fi
  \vadjust{%
    \global\rownum=0 %
  }%
}
\usepackage{multirow}
\usepackage{todonotes}
\usepackage{subcaption}
\usepackage{textcomp}
\usepackage{listings}
\lstdefinestyle{customcpp}{%
  belowcaptionskip=1\baselineskip,
  breaklines=true,
  xleftmargin=\parindent,
  language=C++,
  showstringspaces=false,
  basicstyle=\linespread{0.4}\footnotesize\ttfamily,
  keywordstyle=\bfseries\color{green!40!black},
  numberstyle=\tiny,
  commentstyle=\itshape\color{purple!40!black},
  identifierstyle=\bfseries\color{black},
  stringstyle=\color{red},
  emph={int,char,double,float,unsigned},
  emphstyle=\color{blue},
  morekeywords={uint64_t,uint32_t,__m256i,__m128i,UINT64_C},
}
\lstdefinestyle{custompython}{%
  belowcaptionskip=1\baselineskip,
  breaklines=true,
  xleftmargin=\parindent,
  language=Python,
  showstringspaces=false,
basicstyle=\linespread{0.4}\footnotesize\ttfamily,
  keywordstyle=\bfseries\color{green!40!black},
  numberstyle=\tiny,
  commentstyle=\itshape\color{purple!40!black},
  identifierstyle=\bfseries\color{black},
  stringstyle=\color{red},
  emph={int,char,double,float,unsigned},
  emphstyle=\color{blue},
  morekeywords={uint64_t,uint32_t,__m256i,__m128i,UINT64_C},
}
\usepackage[per-mode=symbol,binary-units=true]{siunitx}

\usepackage[section]{placeins}
\usepackage[ruled]{algorithm}
\usepackage[noend]{algpseudocode}
\definecolor{light-gray}{gray}{0.95}
\usepackage{todonotes}
\usepackage{microtype}
\usepackage[per-mode=symbol]{siunitx}
\usepackage{pgfplots}
\usepackage{tikz}
\usepackage{tikzscale}
\usetikzlibrary{automata,positioning,arrows,quotes}
\definecolor{bblue}{HTML}{4F81BD}
\definecolor{rred}{HTML}{C0504D}
\definecolor{ggreen}{HTML}{9BBB59}
\definecolor{ppurple}{HTML}{9F4C7C}
\definecolor{ggreen}{HTML}{00FF00}
\usepackage{anyfontsize}

\newcommand{\tuple}[2]{(#1_1, #1_2, \ldots, #1_{#2})}  

\papertype{Research Article}

\title{Batched Ranged Random Integer Generation}
\author[1]{Nevin Brackett-Rozinsky}
\author[2]{Daniel Lemire}
\affil[1]{Limerick, Maine, USA

nevin.brackettrozinsky@gmail.com}
\affil[2]{Data Science Research Center, Universit\'e du Qu\'ebec (TELUQ), Montreal, Quebec, H2S 3L5, Canada}

\corraddress{Daniel Lemire, Data Science Research Center, Universit\'e du Qu\'ebec (TELUQ), Montreal, Quebec, H2S 3L5, Canada}
\corremail{daniel.lemire@teluq.ca}

\fundinginfo{Natural Sciences and Engineering Research Council of Canada, Grant Number: RGPIN-2024-03787}
\runningauthor{Nevin Brackett-Rozinsky and Daniel Lemire}
\pgfplotsset{compat=1.18}
\begin{document}

\maketitle
\begin{abstract}
Pseudorandom values are often generated as 64-bit binary words. These random words need to be converted
into ranged values without statistical bias. We present an efficient algorithm to generate multiple independent
uniformly-random bounded integers from a single uniformly-random binary word, without any bias. 
In the common case, our method uses one multiplication and no division operations per value produced.
In practice, our algorithm can more than double the speed of unbiased random shuffling for small to moderately large arrays.
\keywords{Random number generation, Rejection method, Randomized algorithms}
\end{abstract}

\section{Introduction}
\label{sec:introduction}

Random numbers are useful in many computer programs. Most programming languages provide a method to generate uniformly random or pseudorandom integers in the range $[0, 2^L)$ for some $L$, commonly 64. We refer to such numbers as $L$-bit random words, and the functions which produce them as \emph{random number generators}. Pseudorandom generators execute algorithms to produce a sequence of numbers that approximate the properties of random numbers, starting from a given \emph{seed} (typically a few bytes).
There are a variety of efficient techniques to generate high-quality 
pseudorandom binary words, such as the Mersenne Twister~\cite{Matsumoto:1998:MTE:272991.272995},
linear congruential generators~\cite{l1999tables,LEcuyer:1993:SGM:169702.169698,de1988parallelization,fishman2013monte,LEcuyer:1990:RNS:84537.84555},
PCG~\cite{o2014pcg}, and so forth~\cite{l2017random,l2012random}. Although they do not produce truly random outputs, many pseudorandom generators can pass rigorous statistical tests, and are considered to be random in practice~\cite{sleem2020testu01}.

Random binary words are readily available, yet applications often require uniformly random integers from other bounded ranges, such as $[0, b)$ for some $b$. We refer to these numbers as bounded random integers, or \emph{dice rolls}, and we focus on the case where $0 < b \leq 2^L$. We are particularly interested in applications which require multiple dice rolls, and where $b$ may be chosen dynamically rather than known at compile time, such as when shuffling arrays
or selecting random samples~\cite{Vitter:1985:RSR:3147.3165}. For example, we consider the Fisher-Yates random shuffle described by Knuth~\cite{Knuth1969} and by Durstenfeld~\cite{durstenfeld1964algorithm}, which we restate here in \autoref{alg:knuthshuffle}.

\begin{algorithm}
\caption{\centering\strut---
Fisher-Yates random shuffle}
\label{alg:knuthshuffle}
\begin{algorithmic}[1]

\Require Source of uniformly random integers in bounded ranges
\Require Array $A$ made of $n$ elements indexed from $0$ to $n{-}1$
\Ensure All $n!$ permutations of $A$ are equiprobable
\smallskip

\For{$i = n{-}1, \ldots, 1$}
  \State  $j \gets $ random integer in $[0,i]$
  \State exchange $A[i]$ and $A[j]$
\EndFor

\smallskip
\end{algorithmic}
\end{algorithm}

Algorithms such as random shuffling are commonplace in simulations and other important applications~\cite{Devroye:1997:RVG:268403.268413,Calvin:1998:UPR:280265.280273,Owen:1998:LSS:272991.273010,Osogami:2009:FPB:1540530.1540533,Amrein:2011:VIS:1899396.1899401,Hernandez:2012:CNO:2379810.2379813,Hinrichs:2013:SUP:2999611.2999711}. Accordingly, shuffling algorithms are provided in the standard libraries for most programming languages, and there has also been much work done on  parallelizing random permutations~\cite{Shterev2010,Langr:2014:A9P:2684421.2669372,SANDERS1998305,Gustedt2008,Waechter2012}. In the present work we demonstrate an efficient method to roll dice in batches, which can improve the performance of shuffles and similar algorithms by reducing the number of calls to a random number generator. Previous batched dice rolling algorithms were based on division, whereas ours is based on multiplication which can be faster.

Perhaps surprisingly, the computational cost of converting binary words into ranged integers can be critical to good performance. For
example, Lemire~\cite{lemire2019} showed that by simplifying the conversion to avoid most division operations, the practical performance of a random shuffle could be made up to eight times faster. We aim to show that the performance can be further multiplied in some cases, through the use of batched dice rolls.

Our main theoretical result is in \autoref{sec:main_theorem}, where we present and prove the correctness of an algorithm to generate multiple independent bounded random integers from a single random binary word, using in the common case only one multiplication and zero division operations per die roll. The method is based on an existing algorithm due to Lemire that generates one such number at a time, which we summarize in \autoref{sec:existing_algorithms}. Our proof uses mixed-radix notation as described in \autoref{sec:mixed_radix}.

In \autoref{sec:implementation} we demonstrate how our theoretical result can be implemented as a practical algorithm, in \autoref{sec:shuffling_arrays} we apply it to the task of shuffling an array, and in \autoref{sec:experiments} we show the results of our experiments, with timing measurements illustrating the speed of array shuffling with and without batched dice rolls. Finally in \autoref{sec:conclusion} we summarize our results.

\subsection{Mathematical notation}

We only consider non-negative integers. We use ``$\otimes$'' to denote full-width multiplication: $a \otimes b = (x, y)$ means $x$ and $y$ are integers such that $2^L x + y = ab$, and $0 \leq y < 2^L$. On x64 systems, a full-width multiplication requires only a single instruction, while ARM systems provide full-width multiplication with two instructions.

We use ``$\div$'' to denote integer division: $a\div b = \lfloor a/b \rfloor$ is the greatest integer less than or equal to $a/b$. We use ``$\bmod$'' to denote the Euclidean remainder: $(a \bmod b) = a - b \cdot (a \div b)$. Because we use only non-negative integers, we have $0 \leq (a \bmod b) < b$. Note that division and remainder instructions are often slow in practice, and may have a latency of 18~cycles compared to merely 3~cycles for a multiplication on a recent Intel processor (e.g. Ice~Lake)~\cite{abel2019uops}.

We use pi notation ``$\prod$'' to denote products, and we  omit the bounds when they can be inferred from context. Thus if $b_i$ is defined for each $i$ from 1 through $k$, then $\prod b_i = b_1 b_2 \cdots b_k$. We also use sigma notation ``$\sum$'' to denote sums, and an underlined superscript to denote the falling factorial: $n^{\underline k} = n! / (n{-}k)! = n(n-1)\cdots(n-(k{-}1))$.

\section{Existing algorithms}
\label{sec:existing_algorithms}

There are many ways to generate bounded random integers from random bits, and it is nontrivial to do so efficiently.
Several widely-used algorithms are described by Lemire~\cite{lemire2019}, including a then-novel strategy to avoid expensive division operations. This method has now been adopted by several major systems: GNU libstdc++, Microsoft standard C++ library, the Linux kernel, and the standard libraries of the Go, Swift, Julia, C\# and Zig languages.

It works as follows. Let $r$ be a uniformly random integer in $[0, 2^L)$, and $[0, b)$ the desired target range with $0 < b \leq 2^L$. Perform the full-width multiplication $b \otimes r = (x, y)$. Now if $y \geq (2^L \bmod b)$---which we call \emph{Lemire’s criterion}---then $x$ is a uniformly random integer in the range $[0, b)$. Otherwise, try again with a new $r$. That is, apply the rejection method~\cite{von1961various}. A detailed proof of the correctness of Lemire's method can be found in \cite{lemire2019}. We provide a simplified sketch of it here.
\begin{lemma}
\label{lem:lemire}
Lemire's method produces uniformly random integers in the range $[0, b)$.
\end{lemma}

\begin{proof}
Every list of $2^L - (2^L \bmod b)$ consecutive integers contains the same number of multiples of $b$, because the length of that list is divisible by $b$. Thus $br$, which is a uniformly random multiple of $b$ in $[0, 2^L b)$, is equally likely to land in each of the intervals $[2^L x + (2^L \bmod b), 2^L (x+1))$, where $x$ is an integer in  $[0, b)$, because they all have that same length. Lemire's criterion accepts the value $x$ when $br$ is in such an interval, and rejects it otherwise, which ensures that the results are uniformly random.
\end{proof}

In the common case, Lemire’s method uses one $L$-bit random word and one multiplication per die roll. 
Sometimes it requires computing the remainder $(2^L \bmod b)$, which may involve a division instruction. However, the division can often be avoided. Because $(2^L \bmod b) < b$, it is only necessary to compute  $(2^L \bmod b)$ when $y < b$. If $b$ is much smaller than $2^L$ the division is often avoided, and when required it is computed  at most once per die roll. As a practical matter, we note that when using $L$-bit unsigned integers, the value $(2^L \bmod b)$ can be computed in C and similar languages via the expression $-b \mathbin{\%} b$. This works because $-b$ evaluates to $2^L - b$, which is congruent to $2^L \bmod b$.

\subsection{Batched dice rolls}
\label{sec:batched_dice_rolls}

One widely-known way to generate bounded random numbers in batches is that, to simulate rolling dice with~$b_1$ and $b_2$~sides, a single die with $b_1 b_2$ sides is rolled instead. The resulting number is uniformly random in $[0, b_1 b_2)$, and the desired dice rolls can be obtained by taking its remainder and quotient upon dividing by $b_1$. This is equivalent to choosing a random square on a $b_1 \times b_2$ rectangular grid by numbering the squares sequentially, then taking its row and column as the results.

These values are independent and uniformly random integers in $[0, b_1)$ and $[0, b_2)$ respectively, and the method generalizes to more dice by taking successive remainders of quotients upon dividing by each $b_i$. This corresponds to choosing a random cell in a higher-dimensional rectangular lattice by numbering the cells sequentially, and using its coordinates as the results. Or in a context familiar to many programmers, it is equivalent to computing the indices of an element in a multi-dimensional array, at a random position in the underlying linear storage. This strategy uses fewer random bits than rolling each die separately, however its performance is limited by division operations which are slower in practice than multiplications on commodity processors.

When applied to the task of shuffling an array, this approach gives rise to the factorial number system described by Laisant~\cite{Laisant1888} in 1888, and to the Lehmer code~\cite{LehmerCode1960} for numbering permutations. Specifically, dice of sizes $n, n{-}1, \ldots, 2$ can be rolled in a batch by starting with a single $n!$-sided die roll, and extracting the required values through division and remainder. The resulting dice rolls may then be used for a Fisher-Yates shuffle. If $n!$ is too large for practical purposes, several smaller batches may be rolled instead. We are not aware of any widespread use of this batched shuffling method, which we attribute to its reliance on relatively slow division operations.

Our approach to rolling a batch of dice works somewhat differently. Rather than extract the values via division, we instead build them up through multiplication. Specifically, we extend Lemire’s method in order to generate multiple bounded random integers from a single random word, with zero division operations in the common case.

\section{Mixed-radix numbers}
\label{sec:mixed_radix}

We use mixed-radix notation in the proof of our main result. Mixed-radix notation is a positional number system in which each digit position has its own base. This contrasts with decimal notation where every digit is in base 10, or binary where every digit is in base 2. Mixed-radix numbers are allowed, but not required, to have a different base for each digit. Each base is a positive integer, and each digit is a non-negative integer smaller than its base.

We denote both the bases and digits of a mixed-radix number with ordered tuples beginning with the most-significant digit, and use context to distinguish them. For example, a 2-digit mixed-radix number in base $(b_1, b_2)$ whose digits are $(a_1, a_2)$ represents the value $a = a_1  b_2 + a_2$. A $k$-digit mixed-radix number in base $\tuple{b}{k}$ can represent all the integers from 0 through $b{-}1$, where $b = b_1b_2 \cdots b_k$, and this representation is unique. If its digits are $\tuple{a}{k}$ then it represents the value:
\begin{align*}
a = \sum_{i=1}^k \left(a_i \prod_{j>i}^k b_j \right)
= a_1(b_2b_3 \cdots b_k) + a_2(b_3b_4 \cdots b_k) + \cdots + a_{k-1}b_k + a_k
\end{align*}
Given a non-negative integer $a < \prod b_i$, its mixed-radix representation can be obtained by taking the remainders of successive quotients when dividing by $b_k, b_{k-1}, \ldots, b_1$, meaning all the $b_i$ in reverse order.

Mixed-radix notation gives us a concise way to describe the ideas we discussed in \autoref{sec:batched_dice_rolls}. If the digits of a mixed-radix number are interpreted as coordinates in a $k$-dimensional rectangular lattice, then its value matches the label of that cell when they are numbered sequentially. Or if its digits are used as indices in a $k$-dimensional C-style array, then its value gives the corresponding position in the underlying linear storage. We know that a uniformly random cell in a rectangular lattice has independent and uniformly random coordinates, and this translates directly to mixed-radix numbers as seen in \autoref{lem:independence}.
\begin{lemma}\label{lem:independence}
If a uniformly random integer $a$ in $[0, b)$ is written as a mixed-radix number in base $\tuple{b}{k}$, where $b = \prod b_i$, then its digits $\tuple{a}{k}$ are independent and uniformly random integers in the ranges $0 \leq a_i < b_i$.
\end{lemma}

\begin{proof}
Every possible sequence of digits is equally likely, so each $a_i$ takes all integer values in $[0, b_i)$ with uniform probability regardless of the values of the other digits.
\end{proof}

\section{Main result}
\label{sec:main_theorem}

We are now ready to present our main result, which is an efficient method for rolling a batch of dice. We start with a random word, and through repeated multiplications build up the digits of a mixed-radix number. When appropriate conditions are met for that number to be uniformly random, then its digits are uniform and independent.
\begin{theorem}\label{thm:main_theorem}
Let $r_0$ be an $L$-bit random word, meaning $r_0$ is a uniformly random integer in $[0, 2^L)$. Let $\tuple{b}{k}$ be positive integers with $b = \prod b_i \leq 2^L$. Starting with $r_0$ and $b_1$, for each $i$ from 1 to $k$, perform the full-width multiplication $b_i \otimes r_{i-1}$ and set $a_i$ to the most significant $L$~bits and $r_i$~to the least significant $L$~bits of the $2L$-bit result:
\begin{itemize}[label={}]
    \item $(a_1, r_1) \gets b_1 \otimes r_0$
    \item $(a_2, r_2) \gets b_2 \otimes r_1$
    \item $\phantom{(a_1, r_1)} \vdotswithin{\gets}$
    \item $(a_k, r_k) \gets b_k \otimes r_{k-1}$
\end{itemize}
At the end of this process, if $r_k \geq (2^L \bmod b)$ is satisfied, then the $a_i$ are independent and uniformly random integers in the ranges $0 \leq a_i < b_i$.
\end{theorem}

\begin{proof}
We will first show that when the conditions of the theorem are met, if the resulting values $\tuple{a}{k}$ are interpreted as the digits of a mixed-radix number $a$ in base $\tuple{b}{k}$, then $a$ is a uniformly random integer in $[0, b)$. Once this is established, we will invoke \autoref{lem:independence}.

For each $i$ from 1 to $k$, the full-width product $b_i \otimes r_{i-1} = (a_i, r_i)$ means $b_i r_{i-1} = 2^L a_i + r_i$. Since both $r_{i-1}$ and $r_i$ are in $[0, 2^L)$, this implies $0 \leq a_i < b_i$. So each $a_i$ is a valid mixed-radix digit for base $b_i$, and $a$ is well-defined.
Let $c_i$ be the value obtained by truncating $a$ to its first $i$ digits. In other words, $c_i$ is the value represented by the mixed-radix number $\tuple{a}{i}$ in base $\tuple{b}{i}$. Thus $c_1 = a_1$, and $c_i = b_i c_{i-1} + a_i$ for $1 < i \leq k$.

Let $b_1 b_2 \cdots b_i$ be the product of $b_1$ through $b_i$. We claim that $(b_1 b_2 \cdots b_i) \otimes r_0 = (c_i, r_i)$, and we prove it by finite induction on $i$. As a base case, the claim is true for $i = 1$ because $b_1 \otimes r_0 = (a_1, r_1) = (c_1, r_1)$. The claim is equivalent to $r_0 (b_1 b_2 \cdots b_i) = 2^L c_i + r_i$, and when $1 < i \leq k$ we use the inductive hypothesis that this equation holds for $i{-}1$ in order to prove it for $i$:
\begin{align*}
r_0(b_1 b_2 \cdots b_i) &= r_0(b_1 b_2 \cdots b_{i-1}) b_i\\
&= (2^L c_{i-1} + r_{i-1}) b_i
&& \text{(inductive hypothesis)}
\\
&= 2^L b_i c_{i-1} +  b_i r_{i-1}
\\
&= 2^L b_i c_{i-1} +  2^L a_i + r_i
&& \text{(definition of $a_i$ and $r_i$)}
\\
&= 2^L(b_i c_{i-1} + a_i) + r_i
\\
&= 2^L c_i + r_i
&& \text{(formula for $c_i$)}
\end{align*}
This completes the induction and proves the claim for each $i$ from 1 to $k$. We know that $c_k = a$ and $\prod b_i = b$, hence substituting $i \to k$ in the claim gives $b \otimes r_0 = (a, r_k)$. This is the full-width product of $b$ with a random word, and $b \leq 2^L$, so we can apply \hyperref[sec:existing_algorithms]{Lemire’s criterion}~\cite{lemire2019}: if $r_0$ is an $L$-bit random word such that  $r_k \geq (2^L \bmod b)$, then $a$ is a uniformly random integer in $[0, b)$.

By \autoref{lem:independence}, since $a$ is uniformly random in $[0, b)$, its mixed-radix digits in base $\tuple{b}{k}$ are independent and uniformly random in the ranges $[0, b_i)$. But those digits are $\tuple{a}{k}$, so the theorem is proved. Each $a_i$ produced this way is a uniformly random integer in $[0, b_i)$, and the $a_i$ are independent, provided that $r_k \geq (2^L \bmod b)$ and $r_0$ was picked uniformly at random.
\end{proof}

\subsection{Worked example}

\autoref{thm:main_theorem} provides an efficient method to roll a batch of dice, however its description is rather abstract. To develop some familiarity with the process, let us consider a small example using $L = 4$-bit integers. Suppose we want to flip a coin (which is equivalent to rolling a 2-sided die) and also roll a standard 6-sided die. Then $b_1 = 2$, $b_2 = 6$, and $b = 2 \cdot 6 = 12$. Since $12 \leq 16 = 2^4$, this batch of dice can indeed be rolled using 4-bit random numbers. There are $16$ possible values for $r_0$, and we list all of them in \autoref{tab:worked_example}. The columns of the table show the steps in our method.

The first step is to multiply $r_0$ by $b_1$, which yields $2 r_0$ in this example. Then $a_1$ is the number of times 16 goes into that value, and $r_1$ is the remainder. The next step is to multiply $r_1$ by $b_2$, which yields $6 r_1$, and find $a_2$ and $r_2$ in the same way. Finally, Lemire's criterion tells us that if $r_2 < (16 \bmod 12) = 4$, then the batch must be rejected and rolled again. We indicate this with an X in the right-most column of the table.

\begin{table}[htpb]
\caption{Rolling a 2-sided die and a 6-sided die with 4 bits}
\label{tab:worked_example}
\centering
\begin{tabular}{rrrrrrrc}
$r_0$ & $2 r_0$ & $a_1$ & $r_1$ &
$6 r_1$ & $a_2$ & $r_2$ & $r_2 < 4$ \\
\midrule
 0 &  0 & 0 &  0 &  0 & 0 &  0 & X \\
 1 &  2 & 0 &  2 & 12 & 0 & 12 &   \\
 2 &  4 & 0 &  4 & 24 & 1 &  8 &   \\
 3 &  6 & 0 &  6 & 36 & 2 &  4 &   \\
 4 &  8 & 0 &  8 & 48 & 3 &  0 & X \\
 5 & 10 & 0 & 10 & 60 & 3 & 12 &   \\
 6 & 12 & 0 & 12 & 72 & 4 &  8 &   \\
 7 & 14 & 0 & 14 & 84 & 5 &  4 &   \\
 8 & 16 & 1 &  0 &  0 & 0 &  0 & X \\
 9 & 18 & 1 &  2 & 12 & 0 & 12 &   \\
10 & 20 & 1 &  4 & 24 & 1 &  8 &   \\
11 & 22 & 1 &  6 & 36 & 2 &  4 &   \\
12 & 24 & 1 &  8 & 48 & 3 &  0 & X \\
13 & 26 & 1 & 10 & 60 & 3 & 12 &   \\
14 & 28 & 1 & 12 & 72 & 4 &  8 &   \\
15 & 30 & 1 & 14 & 84 & 5 &  4 &
\end{tabular}
\end{table}

The rows with an X in the last column will be rejected, so let us look at the other rows, which produce $a_1$ and $a_2$ as output. We see there are 12 such rows, and we know they are all equally likely because $r_0$ is chosen uniformly at random. There is one such row for each possible pair of values $(a_1, a_2)$ with $a_1$ in \{0, 1\} and $a_2$ in \{0, 1, 2, 3, 4, 5\}, so we conclude that in this example the coin flip and die roll are indeed fair and independent. \autoref{thm:main_theorem} proves that this will always be the case.

\section{Implementation}
\label{sec:implementation}

In some applications the values of $b_i$ are known ahead of time, possibly even at compile time. In that case the value of $t = (2^L \bmod b)$ can be precomputed, and \autoref{thm:main_theorem} can be implemented succinctly as shown in \autoref{alg:known_threshold}. In other applications the values of $b_i$ are not known ahead of time. In that case the threshold $t$ must be computed when needed, which involves a division operation. It can be avoided when $r_k \geq b$, as shown in \autoref{alg:unspecialized}.

\begin{algorithm}
\caption{\centering\strut---
Batched dice rolls (known threshold)}
\label{alg:known_threshold}
\begin{algorithmic}[1]

\Require Source of uniformly random integers in $[0, 2^L)$
\Require Target intervals $[0, b_i)$ for $i$ in $1 \ldots k$, with $1 \leq \prod b_i \leq 2^L$
\Require The value $t = (2^L \bmod \prod b_i)$
\Ensure The $a_i$ are independent and uniformly random in $[0, b_i)$
\smallskip

\Repeat
  \State $r \gets$ random integer in $[0, 2^L)$
  \For{$i$ in $1 \ldots k$}
    \State $(a_i, r) \gets b_i \otimes r$
    \Comment{Full-width multiply}
  \EndFor
\Until{$r \geq t$}

\State \Return $\tuple{a}{k}$

\smallskip
\end{algorithmic}
\end{algorithm}

\begin{algorithm}
\caption{\centering\strut---
Batched dice rolls (unknown threshold)}
\label{alg:unspecialized}
\begin{algorithmic}[1]

\Require Source of uniformly random integers in $[0, 2^L)$
\Require Target intervals $[0, b_i)$ for $i$ in $1 \ldots k$, with $1 \leq \prod b_i \leq 2^L$
\Ensure The $a_i$ are independent and uniformly random in $[0, b_i)$
\smallskip

\State $r \gets$ random integer in $[0, 2^L)$
\For{$i$ in $1 \ldots k$}
  \State $(a_i, r) \gets b_i \otimes r$
  \Comment{Full-width multiply}
\EndFor
\State $b \gets \prod b_i$
\label{line:unspecialized_product}
\If{$r < b$}
  \State $t \gets (2^L \bmod b)$
  \While{$r < t$}
    \State $r \gets$ random integer in $[0, 2^L)$
    \For{$i$ in $1 \ldots k$}
      \State $(a_i, r) \gets b_i \otimes r$
    \EndFor
  \EndWhile
\EndIf
\label{line:unspecialized_endif}

\State \Return $\tuple{a}{k}$

\smallskip
\end{algorithmic}
\end{algorithm}

We must have $b \leq 2^L$ in order to use \autoref{thm:main_theorem}, and for \autoref{alg:unspecialized} we would prefer to have $b$ at least an order of magnitude smaller than $2^L$. If $b$ is too close to $2^L$ then there is a high probability of taking the slow path that needs to calculate $t$, and possibly having to reroll the whole batch of dice. In many applications it is possible to bound the $b_i$ in such a way that a value $u$ satisfying $b \leq u \ll 2^L$ is known ahead of time. This allows for a faster implementation that avoids computing $b$ most of the time, by enclosing lines~\ref*{line:unspecialized_product}--\ref*{line:unspecialized_endif} of \autoref{alg:unspecialized} within an ``if $r < u$'' block.

Whichever version of the algorithm is used, at its core is a loop containing a single full-width multiplication ``$b_i \otimes r$''. The value of $b_i$ is known, but each pass through the loop computes the value of $r$ that will be used for the next iteration. This constitutes a loop-carried dependency, and it means that each iteration must complete before the next can begin.
Modern commodity processors generally have the ability to carry out more than one operation at a time, a feature called instruction-level parallelism or superscalarity~\cite{shen2013modern}. Thus we expect that the computation of  \autoref{alg:unspecialized} can be executed while other operations are completed (e.g. memory loads and stores), or the processor might speculatively use generated values of $a_i$ for upcoming operations. More than one batch of dice rolls may also be interleaved.

\section{Shuffling arrays}
\label{sec:shuffling_arrays}

The Fisher-Yates shuffle of \autoref{alg:knuthshuffle} is widely used for permuting the elements of an array. It requires $n{-}1$ dice rolls to shuffle $n$ elements. With a traditional implementation this involves $n{-}1$ calls to a random number generator, but by rolling dice in batches of $k$ we can reduce that by a factor of $k$. As written, \autoref{alg:unspecialized} requires the computation of the product $b = \prod b_i$. In a batched shuffle, that is $b = n^{\underline{k}}$. One key insight is that we can replace an exact computation by an upper bound $u \geq b$. As long as $r_k \geq u$, there is no need to compute $b$ exactly.

We denote by $n_k$ the largest array length $n$ at which we will use batches of $k$ dice. We want an upper bound $u \geq n^{\underline{k}}$, ideally with $u \ll 2^L$ so the fast path succeeds with high probability. One possible choice is $u_k = {n_k}^{\underline{k}}$, which is the largest product we will ever see for a batch of $k$ dice. We could use $u_k$ as the upper bound for all batches of size $k$, however to improve efficiency we would like to lower the value of $u$ as $n$ decreases. A convenient approach is, whenever the current upper bound fails and we need to compute the true product, we assign that product to $u$ and use it for subsequent batches of size $k$ in the shuffle.

We illustrate this approach in \autoref{alg:partial_shuffle}, which carries out the dice rolls and swaps for a single batch of size $k$, when there are $n$ elements to shuffle. It takes an upper bound $u$ as input, and at the end returns an upper bound for the next iteration. Usually the return value  equals the input, however if the batch needed to calculate its true product then the return value  equals that product.

\begin{algorithm}
\caption{\centering\strut---
Batched partial shuffle}
\label{alg:partial_shuffle}
\begin{algorithmic}[1]

\Require Source of uniformly random integers in $[0, 2^L)$
\Require Array $z$ whose first $n$ elements need to be shuffled
\Require Batch size $k \leq n$ for which $n^{\underline k} \leq 2^L$
\Require Upper bound $u \geq n^{\underline k}$
\Ensure Only the first $(n - k)$ elements of $z$ remain to be shuffled
\smallskip

\State $r \gets$ random integer in $[0, 2^L)$
\For{$i$ in $1 \ldots k$}
  \State $(a_i, r) \gets (n+1-i) \otimes r$
  \Comment{Full-width multiply}
\EndFor
\If{$r < u$}
  \State $u \gets n^{\underline k}$
  \Comment{Falling factorial}
  \State $t \gets (2^L \bmod u)$
  \label{line:start_extra_if}
  \While{$r < t$}
    \State $r \gets$ random integer in $[0, 2^L)$
    \For{$i$ in $1 \ldots k$}
      \State $(a_i, r) \gets (n+1-i) \otimes r$
    \EndFor
  \EndWhile
  \label{line:end_extra_if}
\EndIf
\For{$i$ in $1 \ldots k$}
  \State exchange $z[a_i]$ and $z[n{-}i]$
  \Comment{Zero-based indexing}
\EndFor
\State \Return $u$
\Comment{For the next batch}

\smallskip
\end{algorithmic}
\end{algorithm}

In \autoref{alg:partial_shuffle}, it would be possible to wrap lines~\ref*{line:start_extra_if}--\ref*{line:end_extra_if} in another ``if $r < u$'' block, thus avoiding the computation of $t$ some fraction of the time. We have omitted this for simplicity. We also note that in practice it is rare to need to update $u$ at all, so there is little to gain by further optimizing that scenario.

To shuffle a full array, we first select the largest $k$ such that $n_k \geq n$, and set $u = u_k$. Then we shuffle in batches of $k$, updating $u$ along the way, until $n \leq n_{k+1}$. At that point we set $u = u_{k+1}$ and shuffle in batches of $k{+}1$, and so forth up to some predetermined maximum batch size. Finally, when the number of remaining elements becomes smaller than the previous batch size, we finish the shuffle with one last batch.

\subsection{Batch sizes}
\label{sec:batch_sizes}

The batch size $k$ presents a tradeoff. On one hand, we want to roll as many dice as we can with each random word. On the other hand, we want each batch to succeed on the first try with high probability so we do not have to reroll. These goals are in opposition, and we seek a balance between them. The question becomes, at what array length $n$ should we start rolling dice in batches of $k$. In other words, what values of $n_k$ should be used. This depends on the bit-width~$L$, and can only truly be answered through benchmarks on the target hardware. However, a preliminary analysis can help to identify the right ballpark.

We begin by estimating the computational cost for each die roll. This will enable us to compare the relative performance of different batch sizes $k$ at each array length $n$. If our analysis is successful, then we will be able to use that information to select values of $n_k$ which minimize the expected computational cost for a Fisher-Yates shuffle that uses our batched dice rolling technique of \autoref{alg:partial_shuffle}.

To estimate the cost for each die roll, let us assume that calling the random number generator is as fast as $c_{rng} = 2$ multiplications, and dividing is as slow as $c_{div} = 16$ multiplications. These are conservative estimates, since a fast random number generator reduces the benefit of batching, and a slow division operation increases the cost of rerolling. Using those parameters, and assuming that $u \approx n^{\underline{k}}$, we can find the cost per element for batches of size $k$ at each $n$.

Each batch has probability $p_1 = u / 2^L \approx n^{\underline{k}} / 2^L$ of failing the first threshold, which incurs the cost of computing both $b = n^{\underline{k}}$ and $(2^L \bmod b)$. We treat this as $k{-}1$ multiplications and $1$ division, for a cost of $(c_{div} + k - 1)p_1$. The probability of rerolling the batch is $p_2 = (2^L \bmod b) / 2^L$, so the expected number of rolls including the first is $1 / (1 - p_2)$. Every roll incurs the cost of one random number and $k$ multiplications, for an average of $(c_{rng} + k)/(1 - p_2)$. The batch produces $k$ dice rolls, so the expected cost per die roll is:
\begin{equation*}
c_{avg} = \frac{1}{k}\left(
\frac{c_{rng} + k}{1 - p_2} +
(c_{div} + k - 1)p_1
\right)
\end{equation*}

By choosing $n_k$ so this cost is cheaper with a batch of $k$ than $k{-}1$ for all $n$ up to $n_k$, but not $n_k + 1$, we obtain the values shown in \autoref{tab:batch_sizes}. These are likely to be overestimates for the optimal $n_k$ because we have omitted some of the cost (e.g. misprediction cost). As a rough attempt to adjust for this, we may reduce the values of $n_k$ from \autoref{tab:batch_sizes} to the next-lower power of 2. We use those reduced values in our experiments, as described in \autoref{sec:experiments}.

\begin{table}[htpb]
  \centering
  \caption{Estimated values for $n_k$}
  \label{tab:batch_sizes}
  \begin{tabular}{r|r|r}
    $k$ & $L = 64$ & $L = 32$ \\
    \hline
    2 & \num{1358187913} & \num{20724} \\
    3 &     \num{929104} &   \num{581} \\
    4 &      \num{26573} &   \num{109} \\
    5 &       \num{3225} & \\
    6 &        \num{815} & \\
    7 &        \num{305} & \\
    8 &        \num{146} &
  \end{tabular}
\end{table}

\section{Experiments}
\label{sec:experiments}

We have implemented \autoref{alg:partial_shuffle} in the C programming language, and used it for a batched Fisher-Yates shuffle. To help ensure that our results are reproducible, we make our source code freely available.\footnote{\url{https://github.com/lemire/batched_random}} We use three different random number generators:
\begin{itemize}
    \item Our fastest generator is a  linear congruential generator proposed by Lehmer~\cite{LehmerRNG1951}. It has good statistical properties~\cite{l1999tables}. Using a 128-bit integer seed as the state, we multiply it by the fixed 64-bit integer \texttt{0xda942042e4dd58b5}. The most-significant 64~bits of the resulting state are returned.
    
    \item Our second generator is a 64-bit version of O'Neill's PCG~\cite{o2014pcg}. This relies on a 128-bit parameter $m$ acting as  a multiplier (\texttt{0x2360ed051fc65da44385df649fccf645}). With each call, a 128-bit state variable $s$ is updated: $s \gets m s + c$ where $s$ and $c$ are initialized once. The 64-bit random value is generated from the 128-bit state with a bit rotate and an exclusive~or. We use O'Neill's own implementation adapted to our code base.
    
    \item Finally, we use ChaCha as a 64-bit cryptographically strong generator~\cite{bernstein2008chacha}. The chosen implementation\footnote{\url{https://github.com/nixberg/chacha-rng-c}} is written in conventional C, without advanced optimizations. 
\end{itemize}

We measure the speed of our batched shuffling algorithm across a range of array lengths, with a maximum batch size of either two or six. To get accurate measurements, we shuffle the same array until the total elapsed time is \SI{100}{\micro\second}. Each element of the array occupies 64~bits and the entire arrays fit in CPU cache. We compare the results with two other implementations. One is a standard Fisher-Yates shuffle without any batching, but using an efficient bounded-number algorithm~\cite{lemire2019}. The other is a \emph{naive} batched approach where two dice rolls are extracted from a single random number through division and remainder. Code is provided in~\autoref{sec:code_samples} for our method with batches of up to six, and for both reference methods.

We use two different computer architectures, Apple M2 and Intel Xeon Gold 6338, as shown in~\autoref{tab:test-cpus}. 
We verify that the CPU frequency is stable during the benchmarks: the Intel Xeon Gold 6338's effective frequency remains at \SI{3.2}{\giga\hertz} throughout, while the Apple~M2 remains within 10\% of its 
\SI{3.49}{\GHz} maximum frequency.

\begin{table}[htpb]
\caption{Systems}
\label{tab:test-cpus}
\centering
\begin{tabular}{ccc}
\toprule
Processor &  Intel Xeon Gold 6338 & Apple M2 \\
\midrule
Frequency & \SI{3.2}{\GHz}  & \SI{3.49}{\GHz} \\
Microarchitecture & Ice~Lake (x64, 2019) & Avalanche (aarch64, 2022) \\ 
Memory & DDR4 (3200\,MT/s) &    LPDDR5 (6400\,MT/s)\\
Compiler & LLVM 16  & Apple/LLVM  14  \\
Cache (LLC) &  \SI{48}{\mebi\byte} & \SI{16}{\mebi\byte}\\
\bottomrule
\end{tabular}
\end{table}

For each architecture and random number generator, we plot the resulting average time per element to shuffle an array using all four algorithms. See~\autoref{fig:Lehmer},~\autoref{fig:pcg} and~\autoref{fig:chacha}. Normalizing per element keeps the scale of the results stable, and makes the relative speeds easy to see. As part of our benchmark, we check that the minimal time is close to the average time (within 10\%).

In these graphs, ``shuffle'' is the standard unbatched Fisher-Yates shuffle, and ``naive shuffle\_2'' uses division to roll batches of two dice. The ``shuffle\_2'' and ``shuffle\_6'' methods use our multiplication-based technique to roll batches of at most two and six dice respectively. In particular, shuffle\_6 begins using batches of two when the number of elements to shuffle is at most $2^{30}$ (this is beyond what appears in the graphs), then switches to batches of three when $n \leq 2^{19}$, batches of four when $n \leq 2^{14}$, batches of five when $n \leq 2^{11}$, and batches of six when $n \leq 2^9$.

\begin{figure}[htpb]
 \centering
 \begin{subfigure}[h]{0.49\textwidth}
 \includegraphics[width=0.99\textwidth]{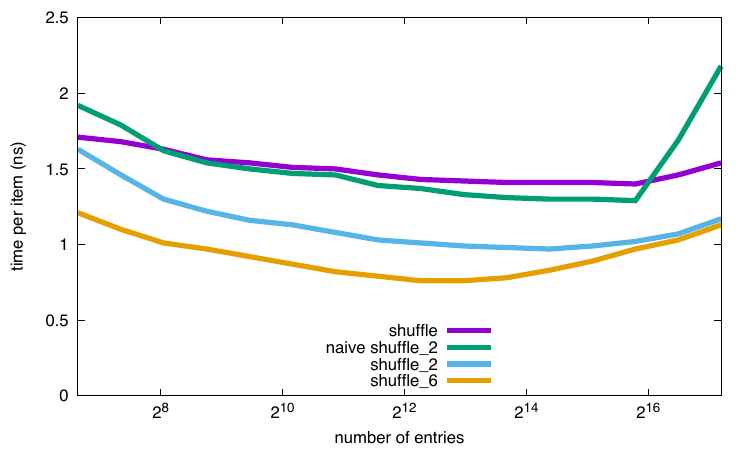}
\caption{Intel Xeon Gold 6338} \end{subfigure}
 \begin{subfigure}[h]{0.49\textwidth}
 \includegraphics[width=0.99\textwidth]{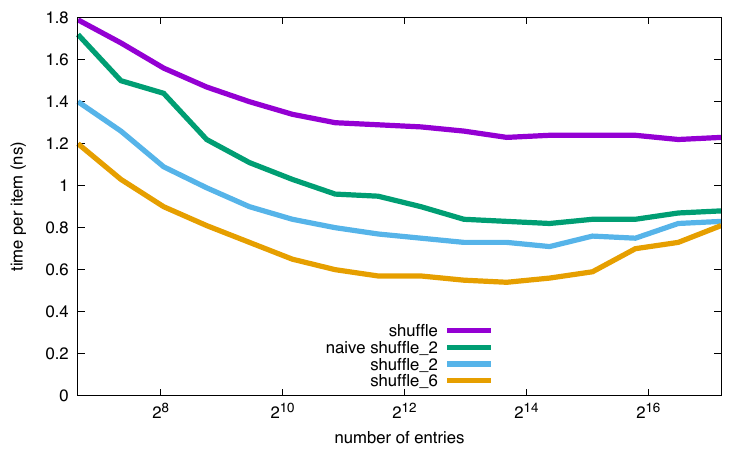}
\caption{Apple M2} \end{subfigure}
\caption{\label{fig:Lehmer}Shuffle timings with Lehmer random number generator}
\end{figure}

\begin{figure}[htpb]
 \centering
 \begin{subfigure}[h]{0.49\textwidth}
 \includegraphics[width=0.99\textwidth]{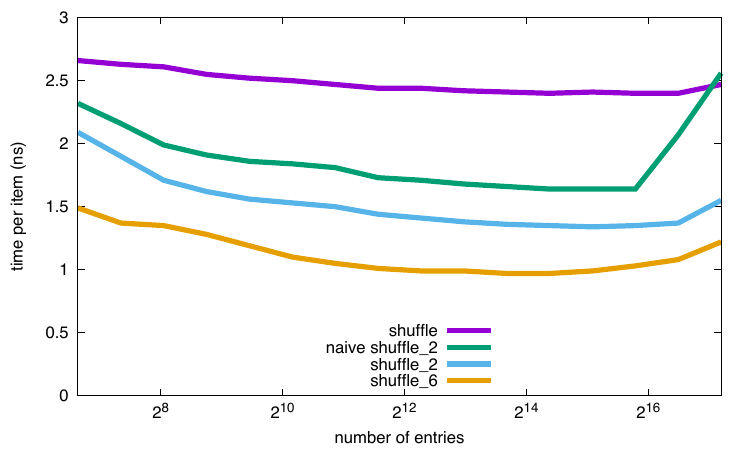}
\caption{Intel Xeon Gold 6338} \end{subfigure}
 \begin{subfigure}[h]{0.49\textwidth}
 \includegraphics[width=0.99\textwidth]{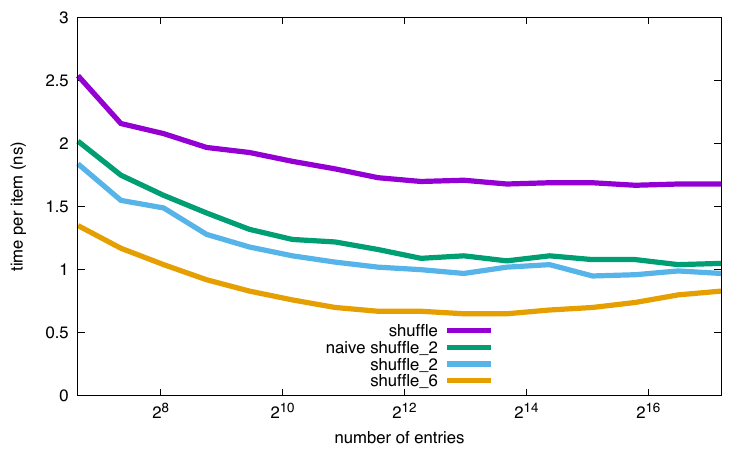}
\caption{Apple M2} \end{subfigure}
\caption{\label{fig:pcg}Shuffle timings with PCG random number generator}
\end{figure}

\begin{figure}[htpb]
 \centering
 \begin{subfigure}[h]{0.49\textwidth}
 \includegraphics[width=0.99\textwidth]{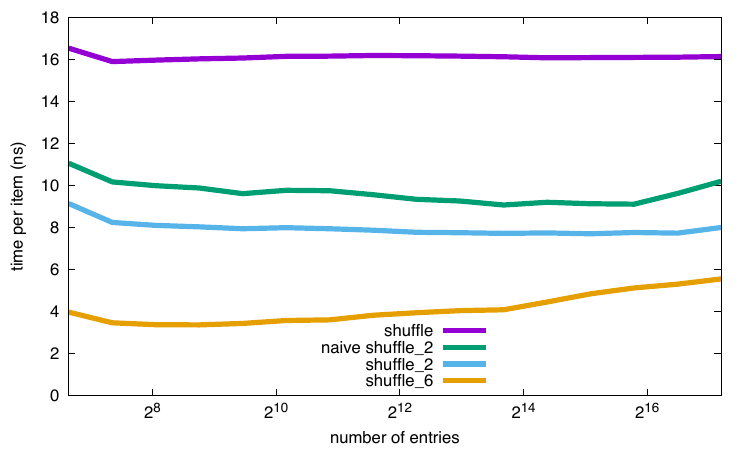}
\caption{Intel Xeon Gold 6338} \end{subfigure}
 \begin{subfigure}[h]{0.49\textwidth}
 \includegraphics[width=0.99\textwidth]{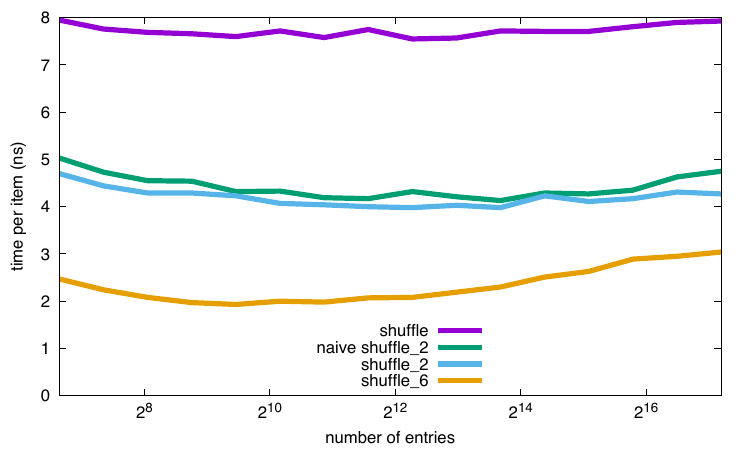}
\caption{Apple M2} \end{subfigure}
\caption{\label{fig:chacha}Shuffle timings with ChaCha random number generator}
\end{figure}

We see in the graphs that using our method for batched dice rolls results in a significantly faster shuffle than either reference method, with larger batches having greater benefits. Across most of the graphs, there is a clear trend: unbatched shuffles are the slowest, batches using division are somewhat faster, and batches using multiplication are the fastest. In one case, using the Lehmer generator on the Intel system, the division-based batched shuffle is no faster than the unbatched shuffle, and sometimes slower. This is explained by the relatively slow division instruction on Intel processors, and the high speed of the Lehmer generator.

In another case, using the ChaCha generator on the Apple system, the division-based method is nearly as fast as the multiplication-based shuffle\_2. This reflects the better performance of division instructions on Apple systems~\cite{applesilicon,fog2016instruction}, and the proportionally greater amount of time taken to generate random numbers with ChaCha, which dominates the time required for the shuffle.

Relative to the unbatched approach, our shuffle\_6 method increases the speed of shuffling by a significant margin across the entire range of sizes we tested, on both platforms, with all three random number generators. This is shown in~\autoref{fig:ratio}, where the y-axis indicates the relative speed of shuffle\_6 compared to the unbatched shuffle.

\begin{figure}[htpb]
 \centering
 \begin{subfigure}[h]{0.49\textwidth}
 \includegraphics[width=0.99\textwidth]{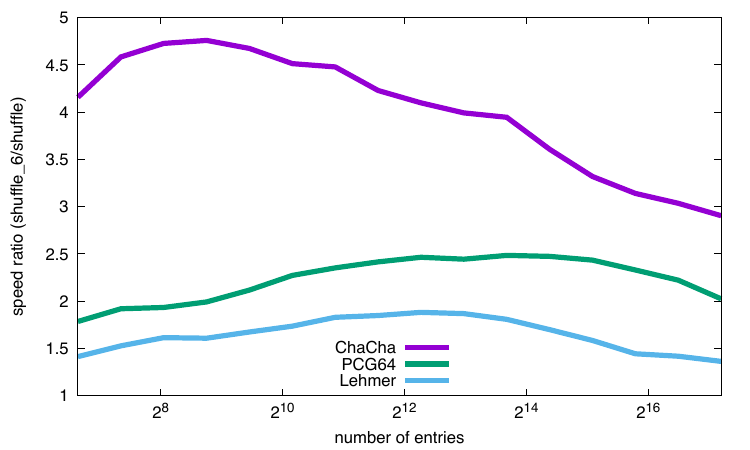}
\caption{Intel Xeon Gold 6338} \end{subfigure}
 \begin{subfigure}[h]{0.49\textwidth}
 \includegraphics[width=0.99\textwidth]{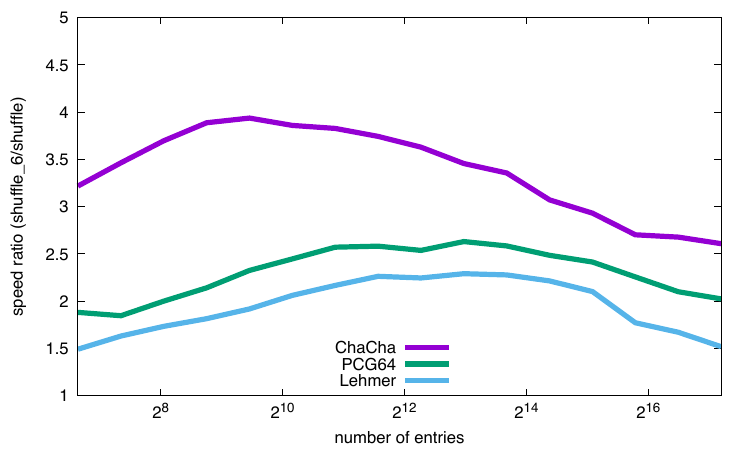}
\caption{Apple M2} \end{subfigure}
\caption{\label{fig:ratio}Speed ratios between shuffle\_6 and a conventional unbatched shuffle}
\end{figure}

On the Intel system, the speedups range from 1.4--1.8$\times$ with the Lehmer generator, 1.8--2.5$\times$ with the PCG generator, and 2.9--4.7$\times$ with ChaCha. On the Apple system, the speedups range from 1.5--2.3$\times$ with the Lehmer generator, 1.9--2.6$\times$ with the PCG generator, and 2.6--3.9$\times$ with ChaCha. Overall on both systems, with a fast generator (e.g. Lehmer or PCG) the speedups are generally 1.5--2.5$\times$, and with a slow generator (e.g. ChaCha) they are 2.5--4.5$\times$. Notably, with the ChaCha random number generator the speedups approach the theoretical ideal of 2$\times$ for batches of two, 3$\times$ for batches of three, and 4$\times$ for batches of four. The faster generators yield proportionally smaller speedups.

All of the graphs in~\autoref{fig:ratio} exhibit a peak somewhere in the middle, with lower values to either side. The decrease at large array sizes is a result of using smaller batches. The decrease at small array sizes is likely due in part to fixed costs being amortized across fewer elements, and it may also represent an opportunity for additional optimization.

\subsection{Instruction counts}

To better understand these results, we use performance counters to record the number of instructions retired by our functions during the shuffling.\footnote{Modern processors execute more instructions than the program flow strictly requires, a process known as \emph{speculative execution}. Among the speculatively executed instructions, only those necessary for the actual program execution flow are considered \emph{retired}.}
The number of instructions does not, by itself, determine the performance because a variety of factors affect the number of instructions that can be executed in a given unit of time. In particular, cache and memory latency limit our maximal speed. Nevertheless, we find that our batched procedures (shuffle\_2 and shuffle\_6) use significantly fewer instructions per element than a conventional unbatched shuffle.
Furthermore, the number of instructions retired per cycle
is maintained, or even increased.
See~\autoref{tab:instructions} where we consider the case when
there are \num{16384}~elements in the array to be shuffled. We also present the results in \autoref{fig:speed}.

\begin{table}[htpb]
\caption{Instructions retired per element and instructions per cycle (\num{16384}~elements)}
\label{tab:instructions}
\centering
\begin{tabular}{llcccc}
\toprule
\multirow{2}{*}{generator} & \multirow{2}{*}{function} & \multicolumn{2}{c}{Intel Xeon Gold 6338} & \multicolumn{2}{c}{Apple~M2}\\
           &          &    ins./el. & ins./cycle &  ins./el. & ins./cycle \\
\midrule
Lehmer
 & shuffle         & 18  & 3.8  &  18  & 4.2  \\
  & naive shuffle\_2 & 14  & 4.4  & 13   & 4.6  \\
 & shuffle\_2      & 13  & 4.6  &  14  & 5.5  \\
 & shuffle\_6      & 10 &  4.0 &  11 & 5.8  \\
PCG64
 & shuffle        & 26 &   3.4  &  24 & 4.1 \\
 & naive shuffle\_2 & 20 & 3.9&  16 &  4.5 \\
 & shuffle\_2      & 19 & 4.6 &  17 & 5.3 \\
 & shuffle\_6      & 12 &  4.0 &  13 & 5.6 \\
ChaCha
 & shuffle       & 139 & 2.8 & 133  &   5.3 \\
 & naive shuffle\_2 & 76 &  3.1 & 71 &   5.1 \\
 & shuffle\_2      & 75 &  3.1 & 72 &   5.3 \\
 & shuffle\_6     &  39 &   3.1&  39 &  5.2 \\
\bottomrule
\end{tabular}
\end{table}

\begin{figure}[htpb]
\centering
 \begin{subfigure}[h]{0.49\textwidth}
 \includegraphics[width=0.99\textwidth]{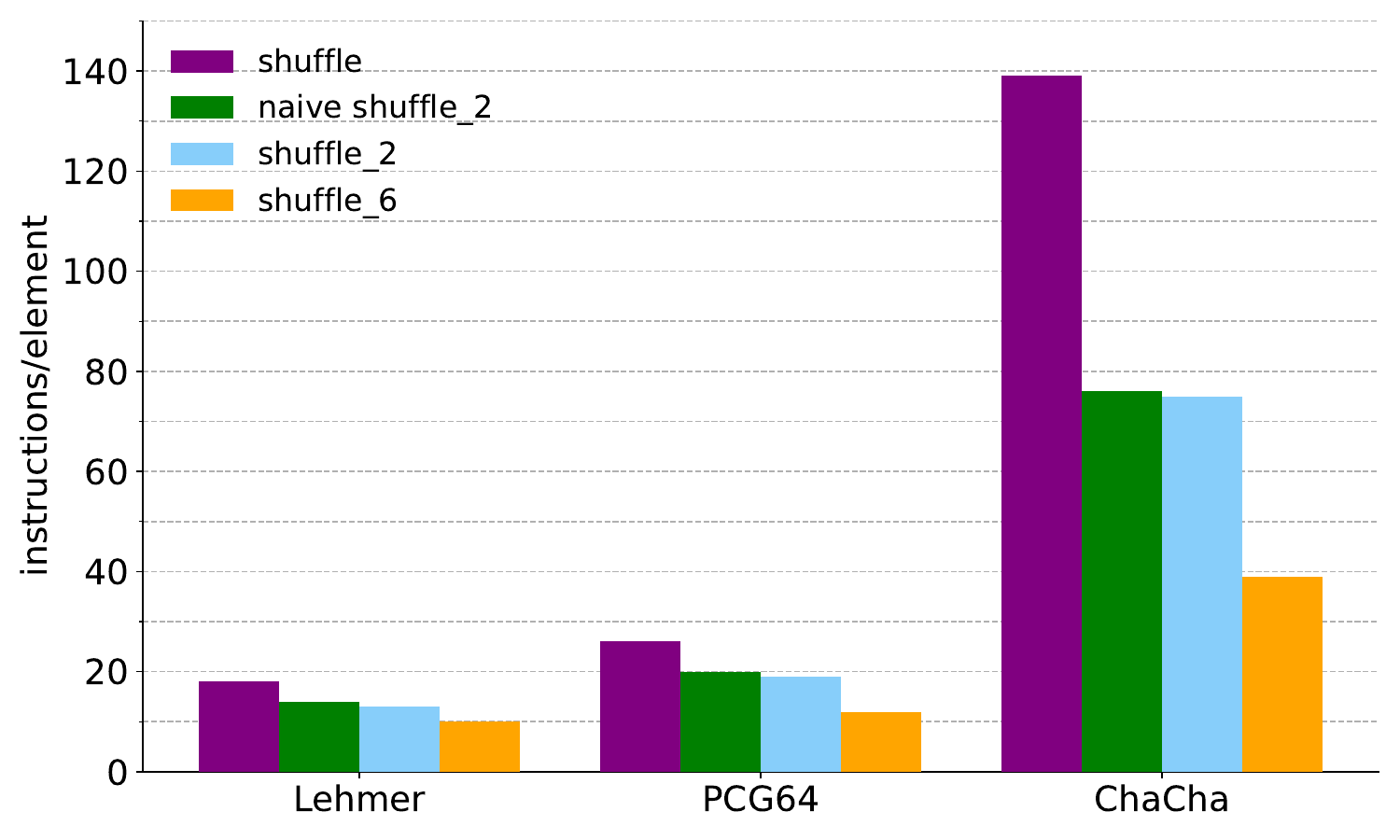}
\caption{Intel Xeon Gold 6338} \end{subfigure}
 \begin{subfigure}[h]{0.49\textwidth}
 \includegraphics[width=0.99\textwidth]{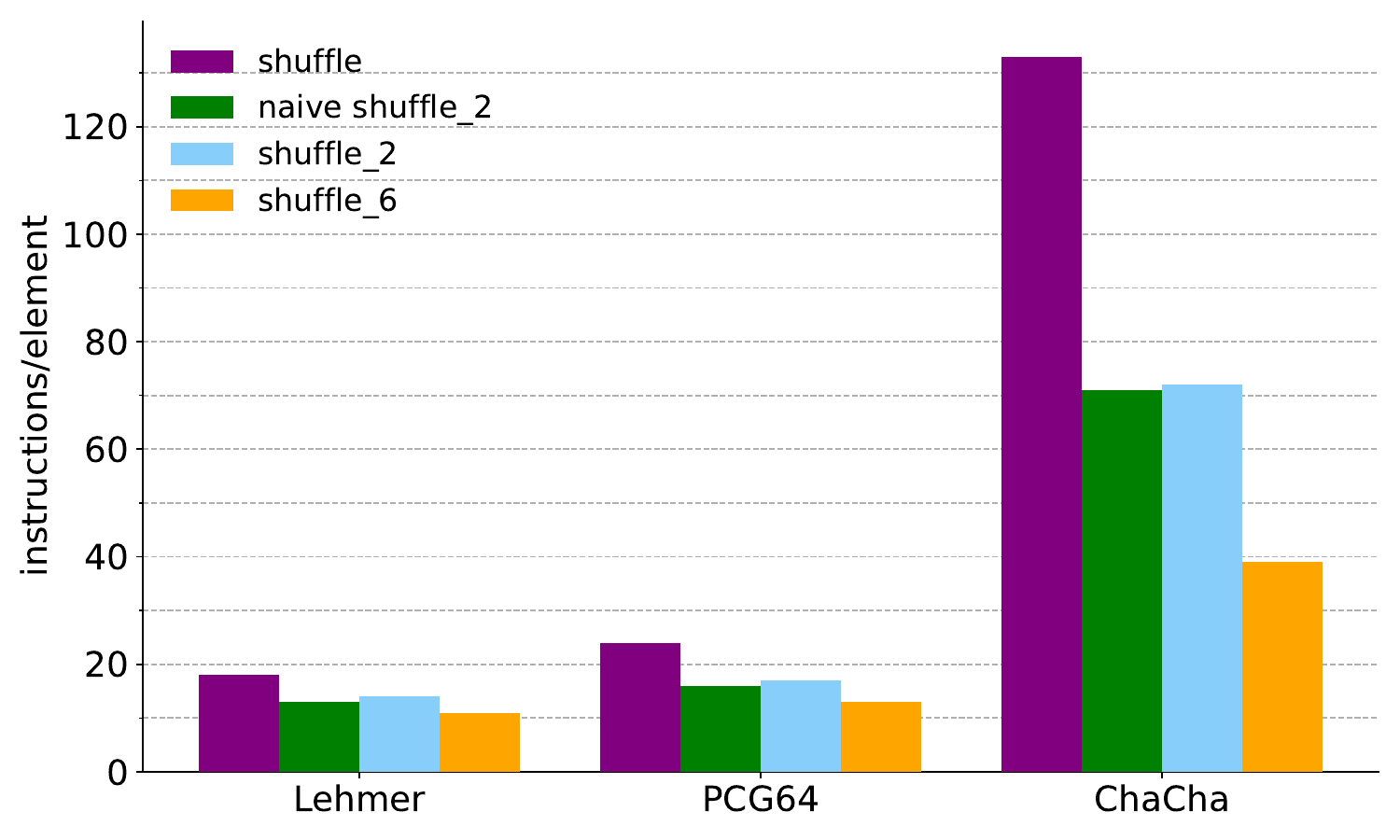}
\caption{Apple M2} \end{subfigure}
\caption{Instructions retired per element for arrays of \num{16384}~64-bit elements\label{fig:speed}}
\end{figure}

The reduction in the number of instructions between a conventional shuffle function and our most aggressively batched one (shuffle\_6) is about 40\% in the case of the Lehmer generator, and rises up to a threefold reduction when using the ChaCha generator. This suggests that the batched shuffles run at a higher speed in large part because they use far fewer instructions. In turn, they use fewer instructions because fewer calls to the random number generator are required.

The division-based approach (naive shuffle\_2) uses about as many instructions as our shuffle\_2 function. However, shuffle\_2 tends to retire more instructions per cycle. This aligns with the fact that multiplication is generally faster than division. It might be possible to further tune the performance of the naive approach, however we expect that division would remain a bottleneck.

\section{Conclusion}
\label{sec:conclusion}

We have shown that Lemire’s nearly-divisionless method of generating bounded random integers can be extended to generate multiple such numbers from a single random word. When rolling several dice or shuffling an array, this batched approach can reduce the number of random bits used, without increasing the amount of computation required. In our benchmark tests we saw speedups by a factor of 1.5--2.5$\times$ when shuffling an array using a fast random number generator (e.g. Lehmer~\cite{l1999tables} or PCG~\cite{o2014pcg}), and by even greater factors when the random bits are more computationally expensive.


Our results are based on a system-oblivious computational model described in  \autoref{sec:shuffling_arrays}. We expect that further tuning, especially system-specific tuning, might have some additional benefits. Our implementation is in the C language using a popular compiler (LLVM). We expect that our results should carry over to other languages such as C++, Rust, Go, Swift, Java, C\# and so forth with relative ease. However, care might be needed to ensure that the generated compiled code is comparable to the result of our C~code under LLVM. We expect that our approach can find broad applications.
Future work should examine other applications such as sampling algorithms and simulations, as well as the shuffling of very large and very small arrays.

\clearpage
\appendix
\section{Code samples}
\label{sec:code_samples}

Unbatched shuffle:
\begin{lstlisting}[style=customcpp]
// generates a bounded random integer
uint64_t random_bounded(uint64_t range, uint64_t (*rng)(void)) {
  __uint128_t random64bit, multiresult;
  uint64_t leftover;
  uint64_t threshold;
  random64bit = rng();
  multiresult = random64bit * range;
  leftover = (uint64_t)multiresult;
  if (leftover < range) {
    threshold = -range % range;
    while (leftover < threshold) {
      random64bit = rng();
      multiresult = random64bit * range;
      leftover = (uint64_t)multiresult;
    }
  }
  return (uint64_t)(multiresult >> 64);
}

// Fisher-Yates shuffle
void shuffle(uint64_t *storage, uint64_t size, 
               uint64_t (*rng)(void)) {
  uint64_t i;
  for (i = size; i > 1; i--) {
    uint64_t nextpos = random_bounded(i, rng);
    uint64_t tmp = storage[i - 1];
    uint64_t val = storage[nextpos];
    storage[i - 1] = val;
    storage[nextpos] = tmp;
  }
}
\end{lstlisting}
Batched shuffle:
\begin{lstlisting}[style=customcpp]
// performs k steps of a shuffle
void batched(uint64_t *storage, uint64_t n, uint64_t k,
               uint64_t bound, uint64_t (*rng)(void)) {
  __uint128_t x;
  uint64_t r = rng();
  uint64_t pos1, pos2;
  uint64_t val1, val2;
  uint64_t indexes[k];
  for (uint64_t i = 0; i < k; i++) {
    x = (__uint128_t)(n - i) * (__uint128_t)r;
    r = (uint64_t)x;
    indexes[i] = (uint64_t)(x >> 64);
  }
  if (r < bound) {
    bound = n;
    for (uint64_t i = 1; i < k; i++) {
      bound *= n - i;
    }
    uint64_t t = -bound % bound;
    while (r < t) {
      r = rng();
      for (uint64_t i = 0; i < k; i++) {
        x = (__uint128_t)(n - i) * (__uint128_t)r;
        r = (uint64_t)x;
        indexes[i] = (uint64_t)(x >> 64);
      }
    }
  }
  for (uint64_t i = 0; i < k; i++) {
    pos1 = n - i - 1; pos2 = indexes[i];
    val1 = storage[pos1]; val2 = storage[pos2];
    storage[pos1] = val2; storage[pos2] = val1;
  }
  return bound;
}

// Fisher-Yates shuffle, rolling up to six dice at a time
void shuffle_6(uint64_t *storage, uint64_t size,
                 uint64_t (*rng)(void)) {
  uint64_t i = size;
  for (; i > 1 << 30; i--) {
    batched(storage, i, 1, i, rng);
  }
  // Batches of 2 for sizes up to 2^30 elements
  uint64_t bound = (uint64_t)1 << 60;
  for (; i > 1 << 19; i -= 2) {
    bound = batched(storage, i, 2, bound, rng);
  }
  // Batches of 3 for sizes up to 2^19 elements
  bound = (uint64_t)1 << 57;
  for (; i > 1 << 14; i -= 3) {
    bound = batched(storage, i, 3, bound, rng);
  }
  // Batches of 4 for sizes up to 2^14 elements
  bound = (uint64_t)1 << 56;
  for (; i > 1 << 11; i -= 4) {
    bound = batched(storage, i, 4, bound, rng);
  }
  // Batches of 5 for sizes up to 2^11 elements
  bound = (uint64_t)1 << 55;
  for (; i > 1 << 9; i -= 5) {
    bound = batched(storage, i, 5, bound, rng);
  }
  // Batches of 6 for sizes up to 2^9 elements
  bound = (uint64_t)1 << 54;
  for (; i > 6; i -= 6) {
    bound = batched(storage, i, 6, bound, rng);
  }
  if (i > 1) {
    batched(storage, i, i - 1, 720, rng);
  }
}
\end{lstlisting}
Naive batched shuffle:
\begin{lstlisting}[style=customcpp]
// performs k steps of a shuffle
void naive_batched(uint64_t *storage, uint64_t n, uint64_t k, uint64_t (*rng)(void)) {
  uint64_t pos1, pos2;
  uint64_t val1, val2;
  uint64_t bound = n;
  for (uint64_t i = 1; i < k; i++) {
    bound *= n - i;
  }
  // Next we generate a random integer in [0, bound)
  uint64_t r = random_bounded(bound, rng);
  for (uint64_t i = 0; i < k - 1; i++) {
    pos2 = r % (n - i);
    r /= (n - i);
    pos1 = n - i - 1;
    val1 = storage[pos1];
    val2 = storage[pos2];
    storage[pos1] = val2;
    storage[pos2] = val1;
  }
  pos2 = r;
  pos1 = n - k;
  val1 = storage[pos1];
  val2 = storage[pos2];
  storage[pos1] = val2;
  storage[pos2] = val1;
}

void naive_shuffle_2(uint64_t *storage, uint64_t size,
                         uint64_t (*rng)(void)) {
  uint64_t i = size;
  for (; i > 0x100000000; i--) {
    naive_batched(storage, i, 1, rng);
  }
  for (; i > 1; i -= 2) {
    naive_batched(storage, i, 2, rng);
  }
}
\end{lstlisting}

\clearpage
\bibliography{references}

\end{document}